\documentclass[12pt]{article}
\usepackage{geometry}

\usepackage{subcaption}

\usepackage[linesnumbered,ruled,vlined]{algorithm2e}
\usepackage{graphicx,xcolor}
\usepackage{amsmath,amsthm,amssymb,bm}
\usepackage[round]{natbib}
\usepackage{enumitem}
\usepackage{multirow}
\usepackage{placeins}
\newcommand{\suit}[1]{\left( #1\right)}

\newcommand{\set}[1]{\left\{ #1\right\}}
\newcommand{\abs}[1]{\left| #1\right|}
\newcommand{\floor}[1]{\lfloor #1 \rfloor}

\newcommand{\mI}{\mathcal{I}}

\newcommand{\bE}{\mathbb{E}}

\newcommand{\mbX}{\boldsymbol{X}}

\newcommand{\hattau}{\widehat{\tau}}
\newcommand{\hatI}{\widehat{I}}

\theoremstyle{plain}
\newtheorem{lemma}{Lemma}

\newtheorem{theorem}{Theorem}

\newtheorem{remark}{Remark}
\title{Clustering Tails in High Dimension}

\author{
    \text{Liujun Chen}\thanks{Faculty of Business for Science and Technology, School of Management, University of Science and Technology of China,   ljchen22@ustc.edu.cn}, \
    \text{Marco Oesting}\thanks{Stuttgart Center for Simulation Science and Institute of Stochastics and Applications, University of Stuttgart, marco.oesting@mathematik.uni-stuttgart.de}, \
   \text{Chen Zhou}\thanks{ Econometric Institute, Erasmus University Rotterdam, zhou@ese.eur.nl}
 }

\begin{document}

\maketitle

\begin{abstract}

One potential solution to combat the scarcity of tail observations in extreme value analysis is to integrate information from multiple datasets sharing similar tail properties, for instance, a common extreme value index. In other words, for a multivariate dataset, we intend to group dimensions into clusters first, before applying any pooling techniques. This paper addresses the clustering problem for a high dimensional dataset, according to their extreme value indices.

We propose an iterative clustering procedure that sequentially partitions the variables into groups, ordered from the heaviest-tailed to the lightest-tailed distributions. At each step, our method identifies and extracts a group of variables that share the highest extreme value index among the remaining ones. This approach differs fundamentally from conventional clustering methods such as using pre-estimated extreme value indices in a two-step clustering method.

We show the consistency property of the proposed algorithm and demonstrate its finite-sample performance using a simulation study and a real data application.

\end{abstract}


\sloppypar

\section{Introduction}
Integrating information from multiple datasets has emerged as a powerful strategy for improving statistical efficiency and predictive accuracy. For example, combining air pollution data across cities \citep{liang2016pm2} or pooling electronic health records from multiple hospitals \citep{cai2022individual} can yield more comprehensive insights than analyzing each source in isolation.  This principle underpins modern methodologies such as distributed learning \citep{battey2018distributed,volgushev2019distributed}, transfer learning \citep{li2022transfer,tian2023transfer}, and multi-task learning \citep{li2025multi,xu2025representation}, all of which aim to leverage shared information across sources or tasks to enhance estimation accuracy. 

One area where data integration is particularly valuable is extreme value analysis, which focuses on understanding the behavior of rare but impactful events \citep{beirlant2004statistics,haan2006extreme,resnick2008extreme}. Consider the Peaks-Over-Threshold (POT) approach in extreme value statistics \citep{davison1990models}. This approach models exceedances over a high threshold by the generalized Pareto distribution. Its statistical counterpart uses selected observations in the tail area, typically a few high order statistics. The total number of selected observations in estimation is often much lower than the sample size, leading to high estimation variance. Therefore, extreme value statistics can benefit largely from integrating tail information across sources, which enhances the stability and reliability of statistical inference for extreme events. Recent studies have successfully explored methods for combining data from multiple sources to improve the estimation of the extreme value index \citep{chen2022distributed,einmahl2022spatial, daouia2024optimal,chen2025distributed}.

In this paper, we provide a novel statistical procedure to cluster dimensions into groups based on their tail behaviour. Existing approaches for the combination of tail information across multiple datasets often rely on the maintained assumption that all datasets share similar tail properties. For instance, in the aforementioned studies to estimate the extreme value index, it is necessary to assume that multiple datasets share a common extreme value index. However, this assumption is often violated in practice, particularly when the data come from heterogeneous sources. In such cases, naively pooling the data can lead to negative transfer \citep{zhang2022survey}, resulting in biased estimation and unreliable inference. This brings us to a fundamental question: how can we group datasets according to their extreme value indices, so that only those with sufficiently similar tail behavior are combined for joint analysis?  To address this, we develop a clustering method that partitions random variables into groups with common extreme value indices. Our approach ensures the validity of tail integration, providing a foundation for extreme value analysis in heterogeneous settings.

Consider independent and identically distributed (i.i.d.)  observations $\mbX_1 = (X_1^{(1)}, \dots, X_1^{(p)}), \dots, \mbX_n = (X_n^{(1)}, \dots, X_n^{(p)})$  drawn from a multivariate distribution function $F$ with marginals $F_1, \dots, F_p$. We assume that the distribution functions $F_j$ are in the max-domain of attraction of  generalized extreme value
distributions with extreme value indices $\gamma_j>0$, i.e., $F_j\sim \mathcal{D}(G_{\gamma_j})$, $j=1,\dots, p$. Mathematically,  
\begin{equation}\label{eq:DoA}
    \lim_{t\to\infty} \frac{1-F_j(tx)}{1-F_j(t)} = x^{-1/\gamma_j}, \quad j=1,\dots,p.
\end{equation}

We further assume that the indices $\set{1,\dots,p}$  can be partitioned into $g$ disjoint groups $\boldsymbol{\tau} = \set{\tau_1, \dots, \tau_g}$ such that $\gamma_i = \gamma_j$ if and only if $i$ and $j$ belong to the same group $\tau_\ell$ for some $\ell \in {1, \dots, g}$. Without loss of generality, we assume that the distinct extreme value indices $\gamma^{(1)}, \dots, \gamma^{(g)}$ corresponding to groups $\tau_1, \dots, \tau_g$ are ordered in decreasing 
order, i.e., $\gamma^{(1)}> \cdots > \gamma^{(g)}>0$.

The goal of this paper is to cluster the $p$ random variables according to their extreme value indices. In particular, the method we propose handles the high dimensional situation where $p\to\infty$ as $n\to\infty$.

A natural baseline approach is to first estimate the extreme value indices $\widehat{\gamma}_1, \dots, \widehat{\gamma}_p$ using standard estimators, and then apply an off-the-shelf clustering algorithm such as the $k$-means method to group the variables based on these estimates \citep{de2023similarity, cao2024tail, wang2025clustering}. However, this two-step procedure has notable limitations. First, accurately estimating the extreme value index for each variable can be highly challenging, especially when the number of extreme observations per variable is limited. Estimation noise may lead to unstable or misleading clustering outcomes. Secondly, estimating the (common) extreme value index is a downstream task after solving the clustering problem. It turns to be a circular problem, had we relied on inaccurate marginal estimations in the clustering step. Thirdly, many standard clustering algorithms such as the $k$-means method require the number of clusters $g$ to be specified in advance, which is often unavailable or difficult to determine in practice.

Our proposed clustering algorithm addresses all these challenges. The algorithm is an iterative clustering procedure that sequentially partitions the $p$ variables into groups, ordered from the heaviest-tailed (highest extreme value index) to the lightest-tailed (lowest extreme value index). At each step, our method identifies and extracts a group of variables that share the highest extreme value index among the remaining ones. 
This approach differs fundamentally from the aforementioned baseline clustering methods in several important aspects. First, 
our procedure does not rely on any assumptions regarding the dependence structure among the variables. The variables can have potential tail dependence among them. Second, it does not require prior knowledge of the number of clusters $g$, but instead determines the grouping structure adaptively from the data. Third, 
  it is readily applicable in high-dimensional settings, where both the number of variables and the scarcity of extreme observations pose significant challenges to conventional approaches. Lastly, it is computationally efficient, avoiding any time-consuming optimization steps and relying entirely on simple rank-based comparisons and thresholding. 

Our work contributes to the growing literature on high-dimensional extreme value analysis. Recent studies in this area have predominantly focused on modeling the tail dependence structure of high-dimensional random vectors \citep{engelke2021learning, engelke2022structure, lederer2023extremes, wan2023graphical}. These approaches often begin by standardizing the marginal distributions and then impose sparsity assumptions on the tail dependence. In contrast, our work shifts the focus to the marginal tail behavior of high-dimensional random vectors. Notably, our approach does not require any assumptions on the sparsity or structure of the tail dependence, thereby broadening the applicability of extreme value methods in high dimensions.

The rest of the paper is organized as follows. In Section 2, we introduce the proposed algorithm and establish its theoretical properties. Section 3 presents a simulation study to examine the finite-sample performance of our estimator. Section 4 demonstrates the practical utility of our method through a real data application. Proofs are collected in Section 5 and partly deferred to the Appendix,  along with some additional simulation results.

Throughout the paper, $\lfloor x \rfloor$ denotes the largest integer less than or equal to $x$.  For a set $A$, $|A|$ denotes the cardinality of $A$. 

\section{Methodology}
\subsection{Number of groups $g$ is known}
Our clustering algorithm follows a two-stage procedure. In the first stage, we standardize the marginal tails to ensure that extreme value indices are comparable across all random variables. In the second stage, we iteratively group the variables based on similarities in their tail behavior.

To compare extreme value indices across dimensions, we first self-scale the observations.
For each random variable $X^{(j)}$, $j=1,\dots,p$, we define the scaled observations as
$$
Y_{i}^{(j)} = \frac{X_{i}^{(j)}}{X_{n-k^*:n}^{(j)}}, \quad\text{for} \quad  i = 1,\dots,n, 
$$
where $X_{1:n}^{(j)} \le \cdots \le X_{n:n}^{(j)}$ are the order statistics of $X_{1}^{(j)},\dots, X_{n}^{(j)}$.  Here, $k^*$ is an intermediate sequence satisfying $k^*\to \infty$ and $k^*/n\to 0$ as $n\to\infty$. To ensure that the normalization is well-defined, $k^*$ should be chosen such that $\min_{1\le j\le p}X_{n-k^*:n}^{(j)}>0$.

Under condition \eqref{eq:DoA}, the exceedances of $X^{(j)}$ over the threshold $X_{n-k^*,n}^{(j)}$ approximately follow a generalized Pareto distribution. With this normalization, we obtain the following approximation for the high order statistics of $Y^{(j)}$, 
$$
Y_{n-k^*x,n}^{(j)} = \frac{X_{n-k^*x,n}^{(j)}}{X_{n-k^*,n}^{(j)}}=x^{\gamma_j}\suit{1+o_P(1)},  \quad \text{as} \ n\to\infty, 
$$
where  $Y_{1,n}^{(j)}\le \cdots\le Y_{n,n}^{(j)}$ denote the order statistics of $Y_1^{(j)},\dots, Y_n^{(j)}$, $j=1,\dots,p$. 
This transformation effectively eliminates the influence of the scale difference, enabling the analysis to focus exclusively on the extreme value index $\gamma_j$. Such self-scaling is also used in the construction of estimators in extreme value statistics, such as  
the Hill estimator \citep{hill1975simple}.

The second stage involves an iterative procedure that sequentially identifies groups of variables in order of decreasing tail heaviness. 
At the core of this stage is a comparison between the marginal high quantile and the high quantile of the pooled observations.
  The key intuition is as follows. Consider another intermediate sequence $k$ satisfying $k = k(n)\to\infty$, $k/k^*\to 0$ and $k/n \to 0$ as $n\to\infty$. If all dimensions share the same extreme value index, then the $(1-k/n)-$quantile of the self-scaled observations at each dimension should be comparable to the quantile at the same probability level  when all self-scaled observations are pooled together. Note that this property is irregardless the (tail) dependence structure across the dimensions. When different marginals have different extreme value indices, this property holds for the cluster sharing the common highest extreme value index, i.e., the heaviest tail. By contrast, consider another dimension with a lower extreme value index. Then the $(1-k/n)-$quantile of the self-scaled observations at this dimension will be at a lower order than the quantile obtained from the pooled observations. This is the key intuition for selecting the cluster with the common highest extreme value index. To formalize the idea, consider any fixed constant $\beta\in (0,1)$. We compare the $(1-\beta k/n)-$quantile among the self-scaled observations at each dimension, i.e., $Y^{(j)}_{n-\floor{\beta k},n}$, with the threshold, which is the $(1-k/n)-$quantile of the pooled observations. If and only if dimension $j$ belongs to the cluster sharing the highest extreme value index,  $Y^{(j)}_{n-\floor{\beta k},n}$ would exceed the threshold. This leads to the following algorithm.

We initialize the candidate set as $\widehat{I}_1 = \set{1,\dots, p}$. At each iteration $\ell=1,\dots, g-1$, 
a variable $j\in \widehat{I}_{\ell}$ is assigned to the $\ell$-th group $\widehat{\tau}_{\ell}$ if its quantile $Y^{(j)}_{n-\floor{\beta k},n}$ exceeds a dynamically updated threshold $u_{\ell}$, which is the $(1-k/n)-$quantile of the pooled observations $\set{Y_i^{(j)}, i=1,\dots,n, j\in \widehat{I}_{\ell}}$. 
After each grouping step, the candidate set is updated as $\hatI_{\ell+1} = \hatI_{\ell} \backslash \hattau_{\ell}$, and the algorithm proceeds to the next iteration.
The algorithm terminates after $g-1$ iterations, with the final group $\widehat{\tau}_g$ consisting of all the remaining variables in $\widehat{I}_g$. This sequential approach ensures that groups are identified in order of their tail behavior, from the heaviest to the lightest tails.


\begin{algorithm}[htb]
    \SetAlgoLined
    \caption{Tail Clustering with known $g$.}
    \label{algorithm:known:g}
    \KwIn{Data $\mbX_1,\dots,\mbX_n$,  partition size $g$, parameters $\beta \in (0,1), k, k^*\in \set{1,\dots,n}$}
    Rescale the observations
    $
    Y_{i}^{(j)} = \frac{X_{i}^{(j)}}{X_{n-k^*:n}^{(j)}} 
    $ for $i = 1,\dots,n$, $j = 1,\dots,p$\; 
    Set $\hatI_1=  \set{1,\dots,p}$\;
    \For{$\ell =1,\dots, g-1$ }{

            Define $u_{\ell}$ as the $(k\cdot |\hatI_{\ell}|)$th upper order statistic of the $(n\cdot  |\hatI_{\ell}|)$- dimensional vector $(Y_{i}^{(j)})_{i=1,\dots,n, \ j\in  \hatI_{\ell}}$ \;
           Set $\widehat{\tau}_{\ell}  = \set{j: j\in \hatI_{\ell}, \ Y_{n-\floor{\beta k}:n}^{(j)}\ge u_{\ell} }$ \;
           Set $\hatI_{\ell+1} = \hatI_{\ell} \backslash \hattau_{\ell}$.
    }
    Set $\hattau_{g} = \hatI_g$\; 
    \KwOut{Partition  $\widehat{\boldsymbol{\tau}} =\set{\hattau_1,\dots,\hattau_g}$.} 
\end{algorithm}

The complete procedure is formalized in Algorithm \ref{algorithm:known:g}.  The algorithm takes as input the data $\mbX_1,\dots,\mbX_n$, the number of groups $g$ and parameters 
$\beta, k$ and $k^*$.  To establish the asymptotic properties of the proposed algorithm, we introduce the following assumptions.


\begin{enumerate}[label=(C\arabic*)]
    \setcounter{enumi}{0}
    \item \label{condition:uniform:convergence}   As $t\to\infty$, 
    $$
    \max_{1\le j\le p}  \sup_{x>1} \abs{\frac{1-F_j(tx)}{1-F_j(t)}x^{1/\gamma_j}-1}\to 0.
    $$
\end{enumerate}

Condition \ref{condition:uniform:convergence} is the  typical regular varying  condition in extreme value analysis. In the high dimensional setting, we require the regular varying conditions hold uniformly for $1\le j\le p$. Similar uniform assumptions have been made in the context of high-dimensional extremes; see, for example, \cite{chen2024high}.

We also impose the following conditions on the dimension $p = p(n)$, and the choices of the intermediate sequences $k$ and $k^*$. 
Denote
$$
\Delta=: \min_{1\le \ell \le g-1} \Delta^{(\ell)} \in (0,1),
$$
where 
$$
\Delta^{(\ell)} = \frac{\gamma^{(\ell)} - \gamma^{(\ell+1)}}{\gamma^{(\ell)}}, \ \ell = 1,\dots,g-1.
$$ 
\begin{enumerate}[label=(C\arabic*)]
    \setcounter{enumi}{1}
\item \label{condition:sequence} As $n\to\infty$, 
    \begin{align}
        k\to\infty,\ & k/k^*\to 0, \  k^*/n\to 0, \label{eq:condition:kstar}\\
        &\frac{p}{(k^*/k)^\Delta} \to 0. \label{eq:condition:p}
    \end{align}
 \end{enumerate}
 Condition \eqref{eq:condition:kstar} requires that $k^*$ be chosen larger than $k$, while condition \eqref{eq:condition:p} imposes an upper bound on the dimensionality $p$.

\begin{remark}
    One example for  $p, k^*, k$ satisfying condition \ref{condition:sequence} can be
    given as follows.   
Assume $p \asymp n^{a}$ with $0<a<\Delta$.  Then condition \ref{condition:sequence} holds with 
$k\asymp n^b, k^* \asymp n^c$ with $0<b< 1-a/\Delta$ and $b+a/\Delta<c<1$. 
\end{remark}

The following two theorems establishes the consistency of the clustering procedure in Algorithm 
\ref{algorithm:known:g}.
\begin{theorem}\label{theorem:prob}
Assume that conditions \ref{condition:uniform:convergence} and \ref{condition:sequence} hold.  Moreover, assume that number of groups $g$ is known and the partition $\widehat{\boldsymbol{\tau}}$ is obtained by Algorithm \ref{algorithm:known:g}.   For any $\beta\in (0,1)$, there exist some constants $C_1>0, C_2>0$ such that,
for sufficiently large $n$,  
$$
    \Pr(\hattau_{\ell} = \tau_{\ell}) \ge 1-C_1 p \exp(-C_2 k), 
$$
for each $\ell \in \set{1,\dots, g-1}$, provided that $\hatI_{\ell} = I_{\ell}$, where  $I_{\ell} = \set{j: \gamma_j\le \gamma^{(\ell)}}$.
\end{theorem}

\begin{theorem}\label{Theorem:prob:all}
    Assume that conditions \ref{condition:uniform:convergence} and \ref{condition:sequence} hold.   Moreover, assume that number of groups $g$ is known and the partition $\widehat{\boldsymbol{\tau}}$ is obtained by Algorithm \ref{algorithm:known:g}. Then, for any $\beta\in (0,1)$, there exist some constants $C_1>0, C_2>0$ such that,
    for sufficiently large $n$,   
$$
\Pr(\widehat{\boldsymbol{\tau}} = \boldsymbol{\tau}) \ge 1-C_1 pg \exp(-C_2 k).
$$
If we assume that $\log (p)/k \to 0$     as $n\to\infty$,  then, we have that,     
    $$
    \Pr(\widehat{\boldsymbol{\tau}} = \boldsymbol{\tau})  \to 1. 
    $$
\end{theorem}

\begin{remark}\label{remark:partical:choice}
    By condition \ref{condition:sequence}, $k$ should be chosen relatively small, while $k^*$ should be chosen sufficiently large.  A practical suggestion is to set $k = \lfloor3\log^{1.05} p\rfloor$ and $k^* = \lfloor n_0^{0.98}\rfloor$, where $n_0$ is the minimum number of the positive observation across      $p$-dimensions, i.e., 
     $n_0= \min_{1\le j\le p} \sum_{i=1}^n \mI\suit{X_{i}^{(j)}>0 }. $

    Theoretically, $\beta$ can be chosen as an arbitrary constant in the interval $(0,1)$. However, the choice of $\beta$ involves a trade-off. If $\beta$ is chosen too large, there is a risk of omitting variables that should rightfully belong to $\tau_{\ell}$. Conversely, if $\beta$
    is chosen too small, the procedure may include an excessive number of variables in $\tau_{\ell}$. 
    Upon examining the details of the proofs, a high level of $\beta$ should be chosen when $k^*/k$ is low and $p$ is high. We then propose a practical choice 
    $$
    \beta = \min(2(k^*/k)^{-1}p+0.5,0.9).
    $$
\end{remark}

Theorem \ref{Theorem:prob:all} establishes that the proposed Algorithm \ref{algorithm:known:g} achieves consistent group partitioning under mild regularity conditions when the number of groups
$g$ is known. The consistency result holds without imposing any assumptions about the dependence structure among the $p$ random variables, which is a significant advantage of our methodology. This robustness makes the algorithm especially valuable in applications where the dependence structure may be complex and unknown.

\subsection{Number of groups $g$ is unknown}

If the number of groups $g$ is unknown, we can continuously partition the set $\set{1,\dots,p}$ until it is empty. The detailed procedure is given in  Algorithm 
\ref{algorithm:unknown:g}.
\begin{algorithm}[htbp]
    \SetAlgoLined
    \caption{Tail Clustering with unknown $g$.}
    \label{algorithm:unknown:g}
    \KwIn{Data $\mbX_1,\dots,\mbX_n$,    parameters $\beta \in (0,1), k, k^*\in \set{1,\dots,n}$}
    Rescale the observations
    $
    Y_{i}^{(j)} = \frac{X_{i}^{(j)}}{X_{n-k^*:n}^{(j)}} 
    $ for $i = 1,\dots,n$, $j = 1,\dots,p$\; 
    Set  $\hatI_1=  \set{1,\dots,p}$; $\ell=1$ \;
    \While{$|\hatI_{\ell}|>0$}{
            Define $u_{\ell}$ as the $(k\cdot |\hatI_{\ell}|)$th upper order statistic of the $(n\cdot  |\hatI_{\ell}|)$- dimensional vector $(Y_{i}^{(j)})_{i=1,\dots,n, \ j\in  \hatI_{\ell}}$ \;
           Set $\widehat{\tau}_{\ell}  = \set{j: j\in \hatI_{\ell}, \ Y_{n-\floor{\beta k}:n}^{(j)}\ge u_{\ell} }$ \;
           Set $\hatI_{\ell+1} = \hatI_{\ell} \backslash \hattau_{\ell}$\;
           Set $\ell \gets \ell+1$\; 
    }
    \KwOut{Partition  $\widehat{\boldsymbol{\tau}}=\set{\hattau_1,\hattau_2,\dots}$.} 
\end{algorithm}

The following theorem establishes the consistency of the clustering procedure in Algorithm \ref{algorithm:unknown:g}. 
\begin{theorem}\label{Theorem:prob:all:unknown:g}
    Assume that conditions \ref{condition:uniform:convergence} and \ref{condition:sequence} hold.   Moreover, assume that number of groups $g$ is unknown and the partition $\widehat{\boldsymbol{\tau}}$ is obtained by Algorithm \ref{algorithm:unknown:g}.  Then, for any $\beta\in (0,1)$, there exist some constants $C_1>0, C_2>0$ such that,
    for sufficiently large $n$,   
$$
\Pr(\widehat{\boldsymbol{\tau}} = \boldsymbol{\tau}) \ge 1-C_1 pg \exp(-C_2 k).
$$
If we assume that $\log (p)/k \to 0$     as $n\to\infty$,  then, we have that,     
    $$
    \Pr(\widehat{\boldsymbol{\tau}} = \boldsymbol{\tau})  \to 1. 
    $$
\end{theorem}

\begin{remark}
After clustering the $p$ random variables into $g$ groups, one can estimate the group-level extreme value indices $\gamma^{(\ell)}$, $\ell = 1, \dots, g$, using within-group averages of the Hill estimators \citep{chen2022distributed, wang2025clustering}. \citet{chen2024high} provided an upper bound on the estimation error for such estimators. Our simulation results confirm that the group-based estimator outperforms individual Hill estimators in terms of accuracy.

\citet{daouia2024optimal} further suggested leveraging the tail dependence structure by using weighted averages of the Hill estimators, rather than simple averages. However, the asymptotic theory for such weighted estimators has so far been developed only in finite dimensions.
\end{remark}

\section{Simulation}\label{sec:simulation}

We generate $p = gq$ random variables from $g$ groups, with each group containing $q$ random variables. Without loss of generality, the groups are defined as:
$$\tau_1 =\set{1,\dots,q}, \tau_2 =\set{q+1,\dots,2q}, \dots, \tau_g = \set{q(g-1)+1,\dots, gq}.$$ 
In other words, for each $j=1,\dots, p$, the variable of dimension $j$ belongs to the group $c_j := \lceil j/q \rceil$.
The extreme value index $\gamma^{(\ell)}$ for the $\ell$-th group is $\gamma^{(\ell)} = (1-\Delta)^{\ell-1}$, $ \ell=1,\dots,g$.
The data are generated from the following two models. 

{\noindent \bf Model (A)}.  The random variables $X^{(j)}, j=1,\dots, p$ are generated independently from an absolute Student-t distribution with degree of freedom $1/\gamma_j$.

{\noindent \bf Model (B)}  We generate $(\widetilde{X}^{(1)},\dots, \widetilde{X}^{(p)})$   from a multivariate Cauchy distribution with scale matrix $\Sigma = (\sigma_{ij})_{i,j=1:p}$ where $\sigma_{ij} = 0.5^{|i-j|}$. Then we transform the marginal distribution of $\widetilde{X}_j$ to an absolute Student-t distribution
with degree of freedom $1/\gamma_j$,  by 
$$
X^{(j)} = \abs{\text{St}_{1/\gamma_j}^{-1}\set{\text{St}_1(\widetilde{X}^{(j)})}}, \quad j=1,\dots,p.
$$
 Here $\text{St}_v(\cdot)$ denotes the cumulative distribution function of a Student-t distribution with degree of freedom $v$.

 Throughout the simulation study, we set $n = 2000$.
 In each replication, we obtain that dimension $j$ is clustered to the group $\widehat{c}_j \in \set{1, 2, \dots}$.
Then the accuracy for this replication is defined as $p^{-1}\sum_{j=1}^p \mI\suit{c_j = \widehat{c}_j}$. The average clustering accuracy is computed based on $B=500$ replications.
    We compare our proposed algorithms with the tail
    $k$-means clustering algorithm, which is detailed in Algorithm \ref{algorithm:tail:kmeans}.  
    
    \begin{algorithm}[htbp]
        \SetAlgoLined
        \caption{Tail $k$-means}
        \label{algorithm:tail:kmeans}
        \KwIn{Data $\mbX_1,\dots,\mbX_n$, number of groups $g$,  parameters $k$}
    For $j=1,\dots,p$, estimate each extreme value index $\gamma_j$ using a Hill estimator 
    $$
    \widehat{\gamma}_j = \frac{1}{k}\sum_{0=1}^{k-1} \log X_{n-i,n}^{(j)} - \log X_{n-k,n}^{(j)};      \;
    $$
  
    Apply the standard $k$-means algorithm to the estimated indices  to partition the $p$ random variables into $g$ groups\;
    Reorder the $g$ groups such that the sum of 
  $\widehat{\gamma}_j$'s in group $\ell_1$  is larger  than 
  that 
  in group $\ell_2$ whenever $\ell_1<\ell_2$ \;
        \KwOut{Partition  $\widehat{\boldsymbol{\tau}}=\set{\hattau_1,\hattau_2,\dots,\widehat{\tau}_g}$.} 
    \end{algorithm}

First, we investigate how the dimensionality affects the performance of the proposed algorithm.  We fix $\Delta=0.5$ and set  $k = \lfloor 3\log^{1.05} (p) \rfloor$, $k^* = \lfloor n^{0.98}\rfloor$ and $\beta =  \min(2(k^*/k)^{-1}p+0.5,0.9)$, as suggested in Remark \ref{remark:partical:choice}. 
Figure \ref{fig:q} shows the average clustering accuracy for varying values of $g$ and $q$. We observe that, the proposed algorithm consistently outperforms the tail $k$-means algorithm in terms of clustering accuracy across all values of $g$ and $q$. Moreover, its performance is relatively robust to the choice of $q$, but exhibits greater sensitivity to the number of groups $g$.
 This finding is align with our theoretical results in Theorem \ref{Theorem:prob:all}.

 \begin{figure}[htb]
    \centering
    \includegraphics[width = 1\textwidth]{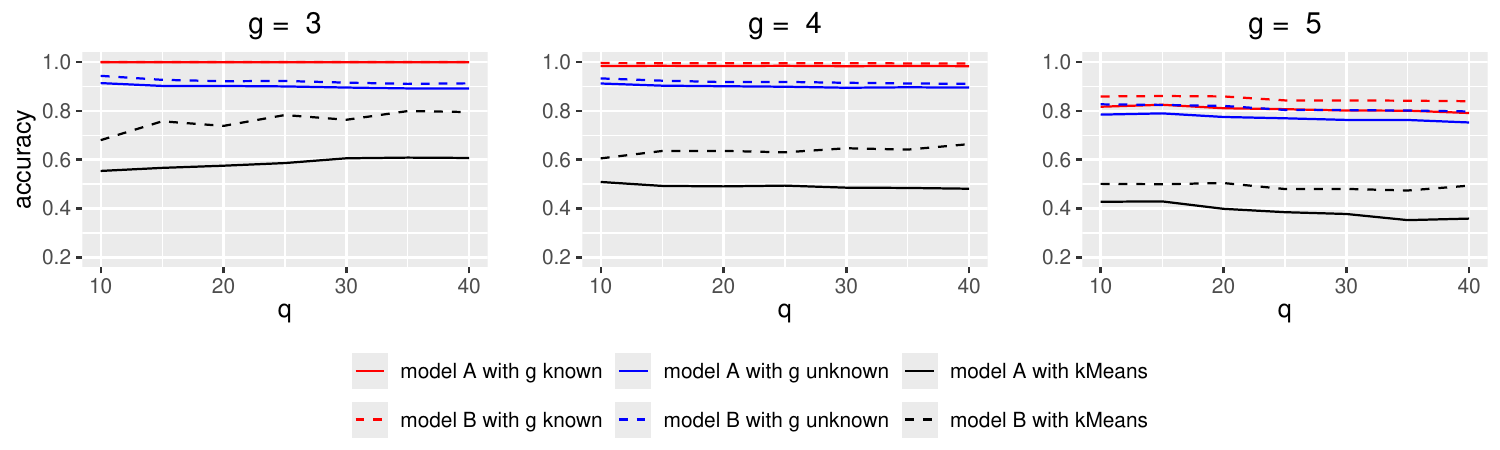}
    \caption{Clustering accuracy  for varying values of $q$ and $g$. }
    \label{fig:q}
\end{figure}

Next,  we  investigate how $\Delta$ affects the performance of the proposed algorithm. We fix $q=15$, $k = \lfloor 3\log^{1.05} (p) \rfloor$, $k^* = \lfloor n^{0.98}\rfloor$ and $\beta =  \min(2(k^*/k)^{-1}p+0.5,0.9)$. 
The averaged clustering accuracy for $g\in \set{3,4,5}$ is plotted against varying values of $\Delta$ in Figure \ref{fig:delta}. 
As $\Delta$ increases, the accuracy of the proposed algorithm initially improves, reaches a peak, and then declines. This behavior can be explained as follows. For low level of  $\Delta$,  the differences in the extreme value indices between groups are insufficient to reliably distinguish the clusters. Conversely, for high level of $\Delta$, $\gamma^{(g)} = (1-\Delta)^{g-1}$  becomes too close to zero, which affects  the performance of  the proposed algorithm based on the assumption of regular variation.

  \begin{figure}[htb]
    \centering
    \includegraphics[width = 1\textwidth]{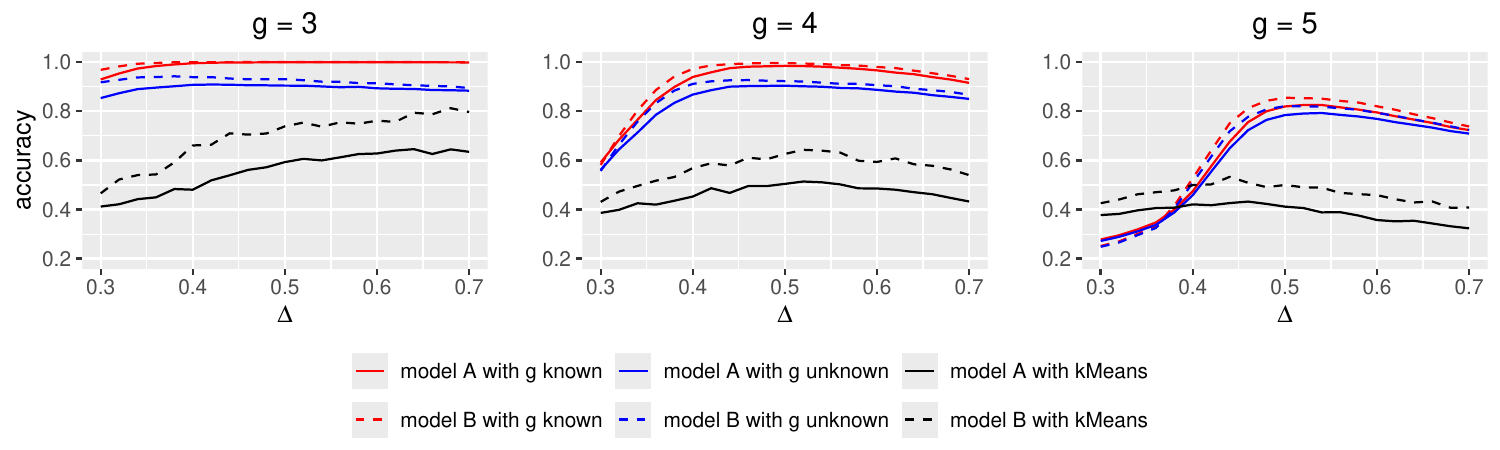}
    \caption{Clustering accuracy  for varying values of $\Delta$ and $g$. }
    \label{fig:delta}
\end{figure}

Finally, we investigate the sensitivity of the proposed algorithm to the parameters $k, k^*$ and $\beta$. Specifically, we evaluate whether the practical choices outlined in Remark \ref{remark:partical:choice} are reasonable. 
To this end, we vary the choice of one parameter while keeping the other two fixed.
 The data are generated with  $\Delta=0.5, q= 15$ and $g\in \set{3,4,5}$. 
 The averaged clustering accuracy is plotted against $k$, $k^*$, and $\beta$ in  Figure \ref{fig:hyper}. 
 The results reveal that the performance of the proposed algorithm is sensitive to the choice of $k$. In particular, a relatively low value of $k$ is required for optimal performance, consistent with the recommendation in Remark \ref{remark:partical:choice}. In contrast, the tail $k$-means algorithm usually requires a relatively high level of $k$ to ensure that $\widehat{\gamma}_j$ is an effective estimator of $\gamma_j$. 
 Regarding $\beta$, the proposed algorithm performs poorly when $\beta$ is too small, while larger values generally lead to better performance—particularly when the number of groups $g$ is large. Moreover, the algorithm is 
 less sensitive to the choice of $k^*$ compared to $k$ and $\beta$.
Numerical results support our practical suggestion of setting \( k^* = \lfloor n^{0.98} \rfloor \), which appears to be a reasonable choice.

\begin{figure}[htbp]
    \centering
    \begin{subfigure}[b]{1\textwidth}
        \includegraphics[width = 1\textwidth]{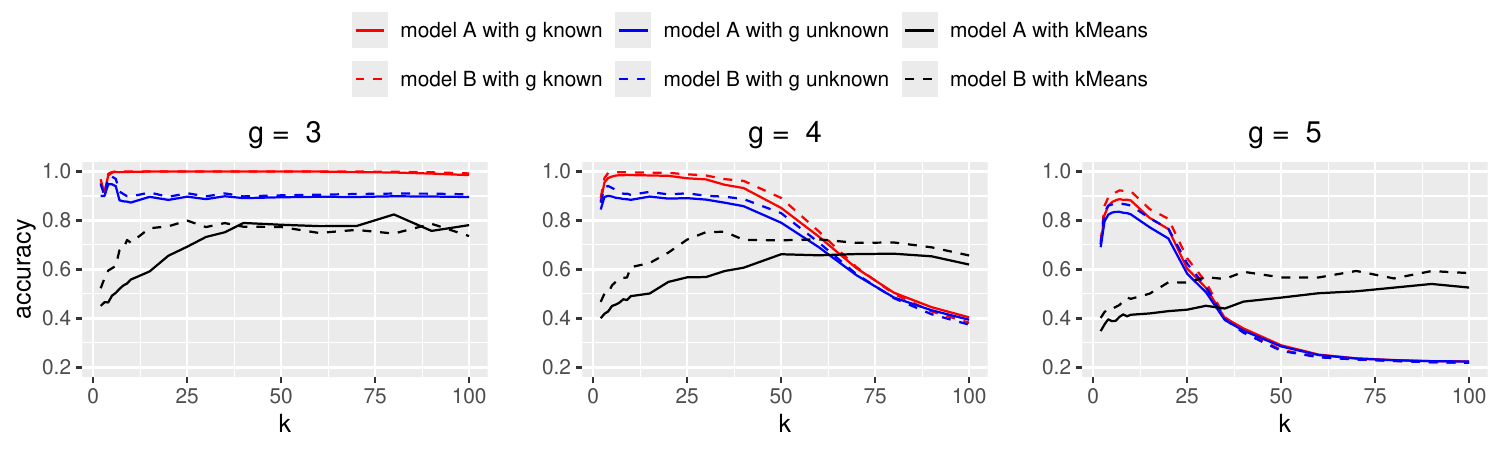}
    \end{subfigure}   
\begin{subfigure}[b]{1\textwidth}
    \includegraphics[width = 1\textwidth]{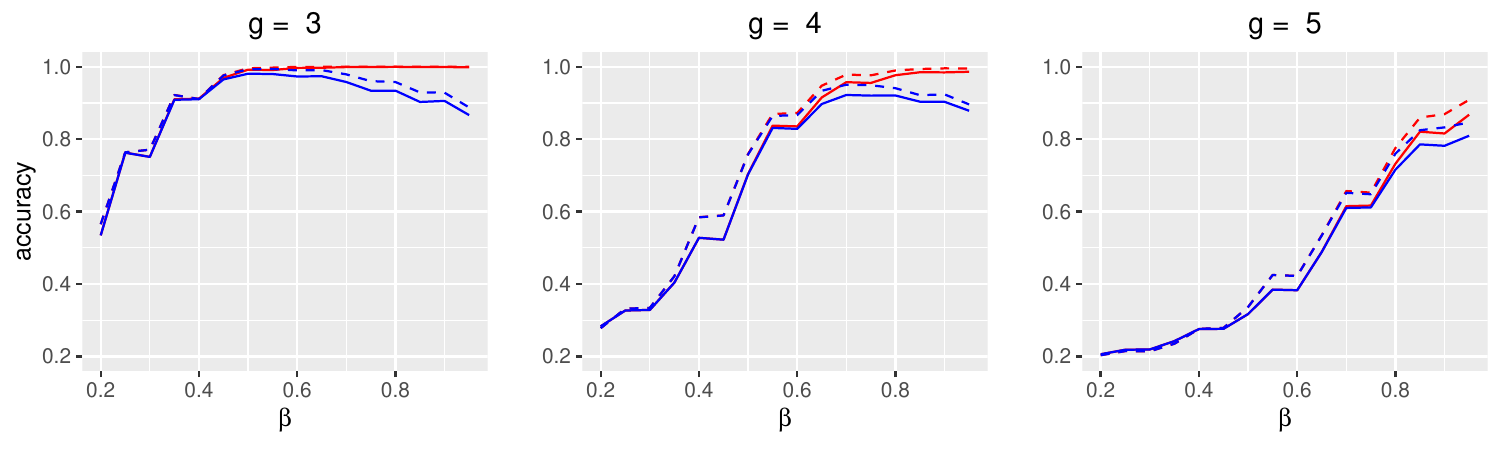}
\end{subfigure}
\begin{subfigure}[b]{1\textwidth}
    \includegraphics[width = 1\textwidth]{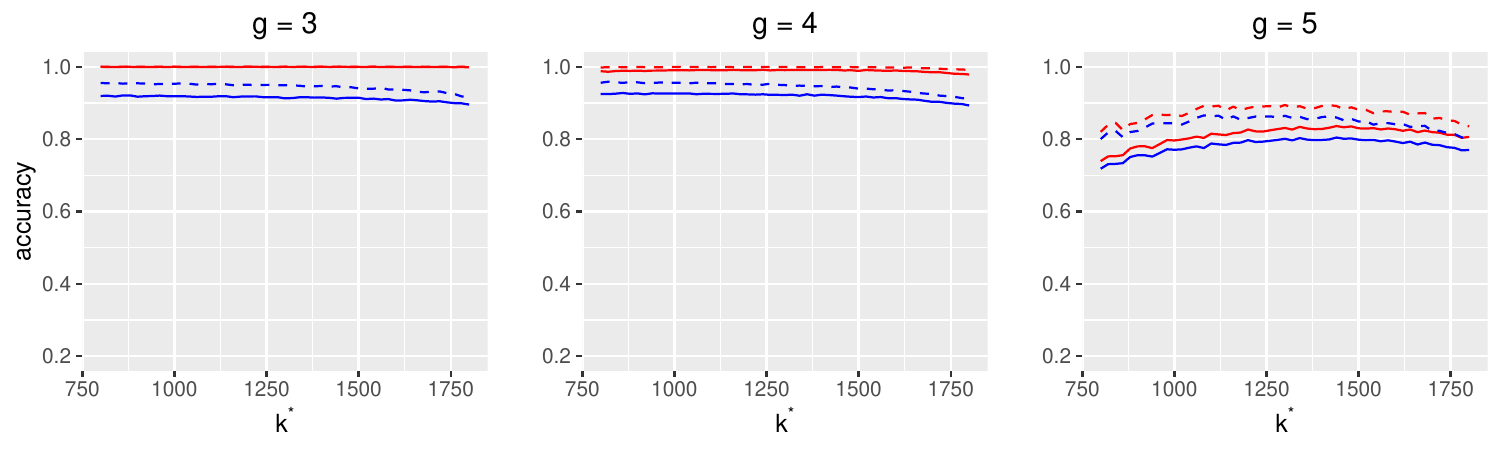}
\end{subfigure}
\caption{Clustering accuracy against different values of $k, \beta$ and $k^*$.}
\label{fig:hyper}
\end{figure}

After clustering the data into $g$ groups $\widehat{\tau}_1,\dots,\widehat{\tau}_g$, we proceed to estimate the extreme value indices by aggregating information within each group.
Specifically, for each group $\ell=1,\dots,g$, we compute the average of the Hill estimates  within that group to obtain a group-level estimate of the extreme value index:
$$
\widetilde{\gamma}^{(\ell)} =\frac{1}{|\widehat{\tau}_{\ell}|} \sum_{j: j\in \widehat{\tau}_{\ell}} \widehat{\gamma}_j, \quad \ell=1,\dots,g.
$$

The final estimate for each individual $\gamma_j$ is then  obtained by assigning the group-level estimate to all members of the group:
$$
\widetilde{\gamma}_j = \sum_{\ell=1}^g \widetilde{\gamma}^{(\ell)}\mI\suit{j\in \widehat{\tau}_{\ell}}, \quad j=1,\dots, p.
$$
We assess the estimation accuracy using the mean squared error (MSE) 
$
p^{-1}\sum_{j=1}^p (\widetilde{\gamma}_j-\gamma_j)^2.
$
We fix $\Delta=0.5$ and set  $k = \lfloor 3\log^{1.05} (p) \rfloor$, $k^* = \lfloor n^{0.98}\rfloor$ and $\beta =  \min(2(k^*/k)^{-1}p+0.5,0.9)$. The average MSE  across the $B=500$ replications 
 for varying values of $g$ and $p$ are shown in Figure \ref{fig:mse:q}.  We observe that the proposed clustering technique significantly reduces the average MSE. In contrast, the tail $k$-means algorithm does not improve the performance of the subsequent estimation process. Furthermore, the results under model (A) are notably better than those under model (B), likely due to the presence of dependence among covariates in model (B). In such cases, simple averaging is suboptimal, and more efficient aggregation may require accounting for the dependence structure through weighted averaging \citep{daouia2024optimal}.

 \begin{figure}[htb]
    \centering
    \includegraphics[width = 1\textwidth]{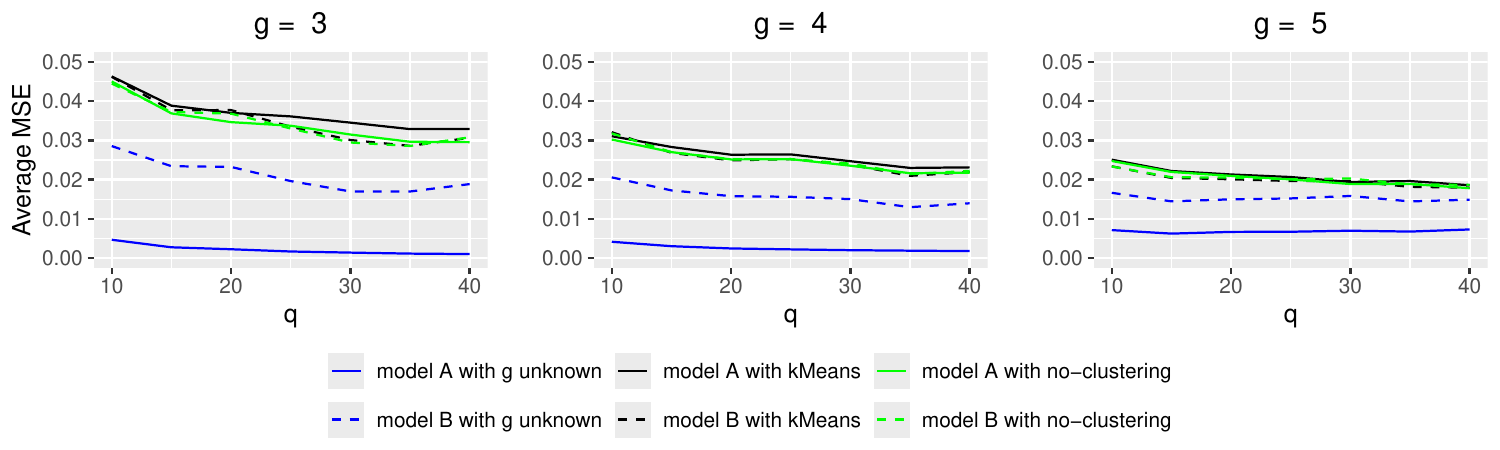}
    \caption{Average MSE for varying values of  $q$ and $g$. }
    \label{fig:mse:q}
\end{figure}

\section{Real data Application}

We demonstrate the practical utility of our proposed algorithm by clustering $p=21$  currency exchange rates (relative to the U.S. dollar) based on daily loss returns (negative log-returns) over the period 2000–2019.
After removing missing observations, the dataset comprises $n=5014$ observations.\footnote{The dataset was downloaded from https://www.kaggle.com/datasets/brunotly/foreign-exchange-rates-per-dollar-20002019.} Following the parameter specifications in Remark \ref{remark:partical:choice}, 
 we set  $k = \lfloor 3\log^{1.05} (p) \rfloor$, $k^* = \lfloor n_0^{0.98}\rfloor$ and $\beta =  \min(2(k^*/k)^{-1}p+0.5,0.9)$. 
 The number of clusters $g$ is inferred using Algorithm \ref{algorithm:unknown:g}, which identifies three distinct groups. 
 
 The first cluster comprises four currencies: CNY (China), MYR (Malaysia), LKR (Sri Lanka), and TWD (Taiwan). These currencies originate from Asian economies and are generally classified as emerging or developing markets. Their grouping likely reflects shared characteristics such as more regulated exchange rate regimes, relatively higher economic volatility, or limited integration with global financial markets.
 
 The second cluster consists of a broader set of currencies, including AUD (Australia), BRL (Brazil), CAD (Canada), INR (India), KRW (South Korea), MXN (Mexico), SGD (Singapore), JPY (Japan), and THB (Thailand). This group spans both emerging and developed economies. Despite their heterogeneity in geographic and developmental context, these currencies may exhibit similar statistical behavior due to shared exposure to global trade flows, commodity markets, or comparable monetary policy environments.

 The third cluster includes EUR (Euro), NZD (New Zealand), GBP (United Kingdom), ZAR (South Africa), DKK (Denmark), NOK (Norway), SEK (Sweden), and CHF (Switzerland).
 Many of these currencies are European or closely integrated with European financial systems. Several, such as CHF, EUR, and GBP, are often considered relatively stable or associated with safe-haven behavior during periods of global financial stress. 
 The inclusion of ZAR and NZD in this group may suggest shared macro-financial attributes with developed economies, such as similar monetary policy regimes or trade relationships, though further analysis would be required to substantiate these connections.

We plot the  Hill estimates of the extreme value indices for these currencies,  along with a 95\% confidence interval in Figure \ref{figure:hill:currencies}. From the figure, we observe that the currencies in Group 1 generally exhibit larger values of the extreme value index compared to those in Groups 2 and 3. Similarly, the currencies in Group 2 tend to have larger extreme value indices than those in Group 3. These findings align with the clustering results.  However, this pattern does not hold uniformly across all individual currencies. This observation highlights a key distinction between our method and the tail $k$-means algorithm, which relies exclusively on the estimated extreme value indices for clustering. 
\begin{figure}[htbp]
    \centering
   \includegraphics[width=1\textwidth]{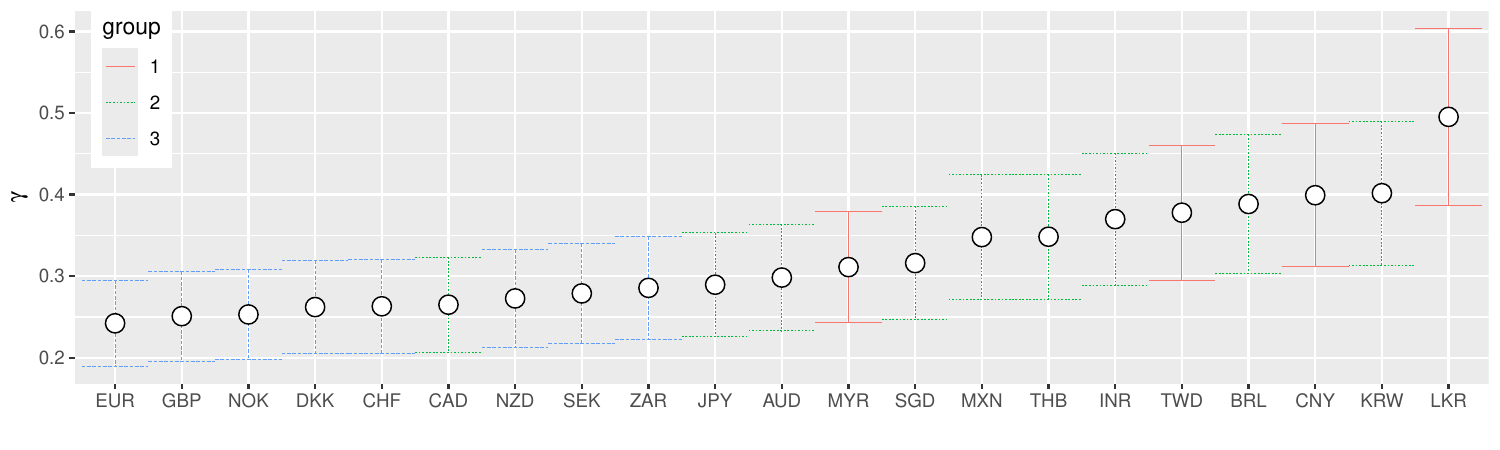} 
   \caption{Hill estimates with upper and lower 95\% confidence limits for the $p=21$ currencies. }
   \label{figure:hill:currencies}
\end{figure}

To assess the robustness of the clustering results, we impose a constraint forcing the currencies to be partitioned into 
$g=2$ groups. Under this restriction, the first group remains unchanged, while the second group combines the currencies originally assigned to the second and third groups.  


\section{Proofs}

\begin{proof}[Proof of Theorem \ref{theorem:prob}]

    By construction, we have that,
   $$
      \max_{j\in \tau_l}  Y_{n- k|I_{\ell}|:n}^{(j)} \le u_{\ell} \le \max_{j\in I_{\ell}} Y_{n-k:n}^{(j)}.
   $$
   Thus, $\hattau_{\ell} = \tau_{\ell}$, provided that, 
   \begin{align}
       \min_{j \in \tau_{l}}Y_{n-\floor{\beta k}:n}^{(j)}\ge &\max_{j\in I_{\ell}} Y_{n-k:n}^{(j)}, \label{s:prob1}\\
       \max_{j \in I_{\ell}\backslash \tau_{\ell}  }Y_{n-\floor{\beta k}:n}^{(j)} < & \max_{j\in \tau_{\ell}} Y_{n-k|I_{\ell}|:n}^{(j)}.  \label{s:prob2}
   \end{align}

Consider any constant $\zeta$ satisfying
   \begin{equation}\label{s:choice:zeta}
    0<\zeta<  \min\set{\frac{\delta}{2(1+\delta)}, 2(1+\beta^{\delta/2})^{-1}-1 },
   \end{equation}
   where $\delta = \min_{1\le j\le p} \gamma_j$, 
   and  a set
   $$
   \mathcal{S} = \set{\abs{\frac{X_{n-k_0:n}^{(j)}}{U_j(n/k_0)} -1} \le \zeta, \quad j \in I_{\ell}, \ k_0\in \set{k^*,k, \floor{\beta k}, k|I_{\ell}|  }   }.
   $$
   We prove  Theorem \ref{theorem:prob} by showing that, for sufficiently large $n$, 
   \begin{itemize}
       \item[(i)]  The relations \eqref{s:prob1} and \eqref{s:prob2} hold on the set $\mathcal{S}$.
       \item[(ii)] The set $\mathcal{S}$ holds with probability larger than $1-C_1p\exp(-C_2 k)$. 
   \end{itemize}
   
{\bf\noindent Proof of  (i)}. 
   On the set $\mathcal{S}$, we have that,
   $$
   \begin{aligned}
      \frac{\min_{j \in \tau_{\ell}} Y_{n-\floor{\beta k}:n}^{(j)} }{ \max_{j\in I_{\ell}} Y_{n-k:n}^{(j)}} =&    \frac{\min_{j \in \tau_{\ell}} \frac{X_{n-\floor{\beta k}:n}^{(j)}}{ X_{n-k^*:n}^{(j)} } }{ \max_{j\in I_{\ell}} \frac{X_{n-k:n}^{(j)}}{X_{n-k^*:n}^{(j)}} }  \\ 
      \ge &\frac{ \min_{j \in \tau_{\ell}} \frac{U_j(n/(\beta k))}{U_j(n/k^*)} }{\max_{j\in I_{\ell}} \frac{U_{j}(n/k)}{U_{j}(n/k^*)} } \suit{\frac{1-\zeta}{1+\zeta}}^2.
   \end{aligned}
   $$
Condition \ref{condition:uniform:convergence} implies that, as $t\to\infty$, 
\begin{equation}\label{s:uniform:U}
    \max_{1\le j\le p}\sup_{x>1}\abs{\frac{U_j(tx)}{U_j(t)}x^{-\gamma_j}-1}\to 0.
\end{equation}
Thus,    for any constant $\varepsilon>0$, there exists $n_1 =n_1(\varepsilon)\ge 1$, such that, for any $n\ge n_1$, 
$$
\begin{aligned}
    (1-\varepsilon) \suit{\frac{k^*}{k\beta}}^{\gamma_j}\le \frac{U_j(n/(\beta k))}{U_j(n/k^*)} \le (1+\varepsilon) \suit{\frac{k^*}{k\beta}}^{\gamma_j},& \quad j\in I_{\ell}, \\
    (1-\varepsilon) \suit{\frac{k^*}{k}}^{\gamma_j}\le  \frac{U_j(n/ k)}{U_j(n/k^*)}\le (1+\varepsilon) \suit{\frac{k^*}{k}}^{\gamma_j}, &\quad j\in I_{\ell}.
\end{aligned}
$$

Then we have 
$$
\begin{aligned}
    \frac{\min_{j \in \tau_{\ell}} Y_{n-\floor{\beta k}:n}^{(j)} }{ \max_{j\in I_{\ell}} Y_{n-k:n}^{(j)}} \ge & \frac{\min_{j \in \tau_{\ell}}  \set{k^*/(k\beta)}^{\gamma_j}  }{\max_{j \in I_{\ell}}  \suit{k^*/k}^{\gamma_j}  }  \suit{\frac{1-\varepsilon}{1+\varepsilon}}\suit{\frac{1-\zeta}{1+\zeta}}^2\\
    =&   \frac{  \set{k^*/(k\beta)}^{\gamma^{(\ell)}}  }{  \suit{k^*/k}^{\gamma^{(\ell)}}  } \suit{\frac{1-\varepsilon}{1+\varepsilon}}\suit{\frac{1-\zeta}{1+\zeta}}^2 \\
    =& \beta^{-\gamma^{(\ell)}} \suit{\frac{1-\varepsilon}{1+\varepsilon}}\suit{\frac{1-\zeta}{1+\zeta}}^2,
\end{aligned}
$$
on the set $\mathcal{S}$.
   By \eqref{s:choice:zeta}, we have that,  
   $$
   \beta^{-\gamma^{(\ell)}} \suit{\frac{1-\zeta}{1+\zeta}}^2>1. 
   $$
   By choosing $\varepsilon = \varepsilon(\zeta,\beta, \min_{1\le\ell\le g}\gamma^{(\ell)})$ sufficiently small, we have that, for  $n\ge n_1$, 
   $$
   \beta^{-\gamma^{(\ell)}} \suit{\frac{1-\varepsilon}{1+\varepsilon}}\suit{\frac{1-\zeta}{1+\zeta}}^2>1, 
   $$
    and hence \eqref{s:prob1} holds.

For \eqref{s:prob2}, on the set $\mathcal{S}$,  
    we have that, 
   $$
   \begin{aligned}
   \frac{\max_{j \in I_{\ell}\backslash \tau_{\ell}  }Y_{n-\floor{\beta k}:n}^{(j)}}{\max_{j\in \tau_{\ell}} Y_{n-k|I_{\ell}|:n}^{(j)}} = & \frac{\max_{j \in I_{\ell}\backslash \tau_{\ell}  }  \frac{X_{n-\floor{\beta k}:n}^{(j)}}{X_{n-k^*:n}^{(j)}}  }{\max_{j\in \tau_{\ell}} \frac{X_{n-k|I_{\ell}|:n}^{(j)}}{X_{n-k^*:n}^{(j)}}  } \\
   \le &      \frac{\max_{j \in I_{\ell}\backslash \tau_{\ell}  } \frac{U_j(n/(\beta k))}{U_j(n/k^*)}   }{\max_{j\in \tau_{\ell}} \frac{U_j(n/(k|I_{\ell}|))   }{U_j(n/k^*)   }}               \suit{\frac{1+\zeta}{1-\zeta}}^2.
   \end{aligned}
   $$
By using the condition \ref{condition:sequence}, we have that, there exists $n_2\ge 1$, such that, for $n\ge n_2$,
$$
\frac{k^*}{k|I_{\ell}|} \ge 1, \quad \frac{k^*}{k\beta} \ge 1.
$$
Thus, by \eqref{s:uniform:U},  we have that, for $n\ge \max(n_1,n_2)$, 
$$
\begin{aligned}
    \suit{\frac{k^*}{k\beta}}^{\gamma_j}(1-\varepsilon) \le  \frac{U_j(n/(\beta k))}{U_j(n/k^*)} \le \suit{\frac{k^*}{k\beta}}^{\gamma_j}(1+\varepsilon), & \quad j \in I_{\ell}\backslash \tau_{\ell}, \\
    \suit{\frac{k^*}{k|I_{\ell}|}}^{\gamma_j}(1-\varepsilon) \le \frac{U_j(n/(k|I_{\ell}|))   }{U_j(n/k^*)   }\le \suit{\frac{k^*}{k|I_{\ell}|}}^{\gamma_j}(1+\varepsilon), & \quad j\in \tau_{\ell}.
\end{aligned}
$$
It follows that, 
$$
\begin{aligned}
    \frac{\max_{j \in I_{\ell}\backslash \tau_{\ell}  } \frac{U_j(n/(\beta k))}{U_j(n/k^*)}   }{\max_{j\in \tau_{\ell}} \frac{U_j(n/(k|I_{\ell}|))   }{U_j(n/k^*)   }}           \le &  \frac{ \max_{j \in I_{\ell}\backslash \tau_{\ell}  }  \set{k^*/(k\beta)}^{\gamma_j} }{ \max_{j\in \tau_{\ell}}  \set{k^*/(k|I_{\ell})|}^{\gamma_j}  } \suit{\frac{1+\varepsilon}{1-\varepsilon}}\suit{\frac{1+\zeta}{1-\zeta}}^2 \\
    =& \frac{\set{k^*/(k\beta)}^{\gamma^{(\ell+1)}}}{\set{k^*/(k|I_{\ell}|)}^{\gamma^{(\ell)}}}    \suit{\frac{1+\varepsilon}{1-\varepsilon}}\suit{\frac{1+\zeta}{1-\zeta}}^2  \\ 
    =& (k^*/k)^{\gamma^{(\ell+1)}-\gamma^{(\ell)}} |I_{\ell}|^{\gamma^{(\ell)}} \beta^{-\gamma^{(\ell+1)}}  \suit{\frac{1+\varepsilon}{1-\varepsilon}}\suit{\frac{1+\zeta}{1-\zeta}}^2   \\ 
    =& \suit{ (k^*/k)^{\frac{\gamma^{(\ell+1)}}{\gamma^{(\ell)}}-1}  |I_{\ell}|}^{\gamma^{(\ell)}} \beta^{-\gamma^{(\ell+1)}}    \suit{\frac{1+\varepsilon}{1-\varepsilon}}\suit{\frac{1+\zeta}{1-\zeta}}^2   \\
    =&  \suit{ (k^*/k)^{-\Delta^{(\ell)}}  |I_{\ell}|}^{\gamma^{(\ell)}} \beta^{-\gamma^{(\ell+1)}}   \suit{\frac{1+\varepsilon}{1-\varepsilon}}\suit{\frac{1+\zeta}{1-\zeta}}^2 .\\
\end{aligned}
$$

By condition \ref{condition:sequence}, we have that, there exists $n_3 =n_3(\varepsilon)\ge 1$, such that for $n\ge n_3$,
$$
\suit{ (k^*/k)^{-\Delta^{(\ell)}}  |I_{\ell}|}^{\gamma^{(\ell)}} <\varepsilon.
$$
By choosing $\varepsilon$ sufficiently small, we have that, for $n\ge \max(n_1,n_2,n_3)$, 
$$
\frac{\max_{j \in I_{\ell}\backslash \tau_{\ell}  } \frac{U_j(n/(\beta k))}{U_j(n/k^*)}   }{\max_{j\in \tau_{\ell}} \frac{U_j(n/(k|I_{\ell}|))   }{U_j(n/k^*)   }}   <1,
$$
on the set $\mathcal{S}$. Consequently,  \eqref{s:prob2} holds.

{\bf\noindent Proof of  (ii)}.    Denote 
   $$
   \begin{aligned}
       \mathcal{S}_1= &\set{\abs{\frac{X_{n-k^*:n}^{(j)}}{U_j(n/k^*)} -1} \le \zeta, \quad j \in I_{\ell} }, \\
       \mathcal{S}_2= &\set{\abs{\frac{X_{n-k:n}^{(j)}}{U_j(n/k)} -1} \le \zeta, \quad j \in I_{\ell} }, \\
       \mathcal{S}_3= &\set{\abs{\frac{X_{n-\floor{\beta k}:n}^{(j)}}{U_j(n/ (\beta k))} -1} \le \zeta, \quad j \in I_{\ell} }, \\
       \mathcal{S}_4= &\set{\abs{\frac{X_{n-k|I_{\ell}|:n}^{(j)}}{U_j(n/(k|I_{\ell}|))} -1} \le \zeta, \quad j \in I_{\ell} }. \\
   \end{aligned}
   $$
   Then, $\mathcal{S} = \cap_{i=1}^4 S_i$. By Lemma \ref{lemma:concentration:order:statistic}, we have that,    for sufficiently large $n$, 
$$
\begin{aligned}
    \Pr(\mathcal{S}_1) =  \Pr\suit{\max_{j\in \mathcal{I}_{\ell}} \abs{\frac{X_{n-k^*:n}^{(j)}}{U_j(n/k^*)} -1} \le \zeta } 
    =& 1-\Pr\suit{\max_{j\in \mathcal{I}_{\ell}} \abs{\frac{X_{n-k^*:n}^{(j)}}{U_j(n/k^*)} -1}> \zeta } \\
    \ge & 1-|I_{\ell}|\max_{j\in \mathcal{I}_{\ell}} \Pr\suit{ \abs{\frac{X_{n-k^*:n}^{(j)}}{U_j(n/k^*)} -1}> \zeta }\\
    \ge & 1- 4|I_{\ell}| \exp\suit{-Ck^* \zeta^2} \\
    \ge  & 1- 4 p\exp\suit{-Ck^* \zeta^2},
\end{aligned}
$$
where 
$$
C = \frac{1}{16\max_{1\le j\le p}\gamma_j^2} .
$$
Similarly, we have that, for sufficiently large $n$,
   $$
   \begin{aligned} 
       \Pr(\mathcal{S}_2) \ge & 1-  4p\exp\suit{-Ck \zeta^2}, \\  
       \Pr(\mathcal{S}_3) \ge &1- 4p\exp\suit{-C\beta k \zeta^2}, \\  
       \Pr(\mathcal{S}_4) \ge &1- 4p\exp\suit{-Ck |I_{\ell}|  \zeta^2}.
   \end{aligned}
   $$
It follows that,   for sufficiently large $n$, 
   $$
   \begin{aligned}
    \Pr(\mathcal{S}) \ge 1- 4p\set{\exp\suit{-Ck^* \zeta^2}+\exp\suit{-Ck\zeta^2}+ \exp\suit{-C\beta k \zeta^2} + \exp\suit{-C k |I_{\ell}|  \zeta^2} } .
   \end{aligned}
   $$
   Note that, $k^*/k\to\infty$  as $n\to\infty$ and $\zeta$ is constant. Thus,  there exist $C_1>0, C_2>0$ such that, for sufficiently large $n$,  
   $$
   \Pr(\mathcal{S}) \ge 1-C_1p\exp\suit{-C_2 k}.
   $$
   The proof is complete.
   \end{proof}
   \begin{proof}[Proof of Theorem \ref{Theorem:prob:all}]
    Define $p_{\ell} = \Pr(\hattau_l = \tau_l |\hatI_l = I_l)$. Note that, 
    $$
     \begin{aligned}
        \Pr\suit{\widehat{\boldsymbol{\tau}} =\boldsymbol{\tau} }=& \Pr(\hattau_1 = \tau_1, \cdots, \hattau_g = \tau_g) \\
         =& \Pr(\hattau_1 = \tau_1|\hatI_1 = I_1) \Pr(\hattau_2 = \tau_2, \cdots, \hattau_g = \tau_g|\hattau_1 = \tau_1,\hatI_1 = I_1 ) \\ 
         =& p_{1} \Pr(\hattau_2 = \tau_2|\hatI_2 = I_2) \Pr(\hattau_3 = \tau_3, \cdots, \hattau_g = \tau_g|\hatI_3 = I_3) \\
         =& \cdots \\
         =& \prod_{\ell=1}^{g-1} p_{\ell}.
     \end{aligned}
    $$
 By Theorem \ref{theorem:prob}, we have that, for sufficiently large $n$, 
 $$
 \begin{aligned}
    \Pr\suit{\widehat{\boldsymbol{\tau}} =\boldsymbol{\tau} } \ge & \prod_{\ell=1}^{g-1} \set{1-C_1 p \exp(-C_2 k )} \\
     \ge & 1-g C_1 p \exp(-C_2 k ) \\
     \ge & 1-  C_1 pg \exp(-C_2 k ) .
 \end{aligned}
 $$
By using the condition $\log (pg)/k \to 0$ as $n\to\infty$,  we have that,
 $$
 C_1 pg \exp(-C_2 k )\to 0.
$$
and hence
$$
\Pr(\hattau = \tau ) \to 1.
$$

 \end{proof}

 \begin{proof}[Proof of Theorem \ref{Theorem:prob:all:unknown:g}]
    By Theorems \ref{theorem:prob} and \ref{Theorem:prob:all}, it remains to show that,   there exist constants $C_1>0, C_2>0$ such that, for sufficiently large $n$, 
   $$
    \Pr\suit{  \min_{j \in I_{g}}Y_{n-\floor{\beta k}:n}^{(j)}> u_{g}}  \ge 1-C_1p\exp(-C_2k).
   $$
Given that
   $u_{g} \le \max_{j \in I_g} Y_{n-k:n}^{(j)}$, 
   we only need to show that,   for sufficiently large $n$,
   $$
   \Pr\suit{  \min_{j \in I_{g}}Y_{n-\floor{\beta k}:n}^{(j)}> \max_{j \in I_g} Y_{n-k:n}^{(j)}}  \ge 1-C_1p\exp(-C_2k).
   $$
   The proof of this result is analogous to the proof of Theorem \ref{theorem:prob} and is thus omitted. 
 \end{proof}

\FloatBarrier
\bibliographystyle{apalike} 
\bibliography{mybib.bib}

\appendix

\setcounter{figure}{0}   
\renewcommand{\thefigure}{S\arabic{figure}}

\section{Some Lemmas}
\setcounter{equation}{0}   
\renewcommand{\theequation}{A\arabic{equation}}
\setcounter{lemma}{0}   
\renewcommand{\thelemma}{A\arabic{lemma}}
\setcounter{prop}{0}   
\renewcommand{\theprop}{A\arabic{prop}}

\begin{lemma}[Bernstein’s inequality, see \cite{shorack1986empirical}, page 855]\label{Lemma:Bernstein}
    Let $Z_1,\dots, Z_n$ be independent random variables with $|Z_i|\le M$ almost surely and $\bE(Z_i) = \mu$ for all $1\le i\le n$, where $M>0$, $\mu\in \mathbb{R}$. Then, for any $\varepsilon>0$,
$$
\Pr\suit{ \abs{\frac{1}{n} \sum_{i=1}^n Z_i -\mu}>\varepsilon} \le 2 \exp\set{-\frac{n\varepsilon^2}{2\frac{1}{n}\sum_{i=1}^n \textnormal{Var}(Z_i)+2M\varepsilon/3}}.
$$
\end{lemma}

\begin{lemma}\label{lemma:concentration:order:statistic}
Assume that condition \ref{condition:uniform:convergence} holds. Let $k =k(n)$ be an intermediate sequence satisfying $k\to\infty, k/n\to 0$ as $n\to\infty$.   Let $\zeta$ be a constant satisfying $0<\zeta\le \delta/(2(1+\delta))$, where $\delta = \min_{1\le j\le p}\gamma_j$. 
 Then,   for sufficiently large $n$, 
    $$
    \Pr\suit{\abs{\frac{X_{n-k,n}^{(j)} }{U_j(n/k)}-1}>\zeta} \le 4 \exp\suit{ -\frac{k\zeta^2}{16\max_{1\le j\le p}\gamma_j^2} }, \quad j=1,\dots, p,
    $$
Here, $U_j(x)=F_j^{\leftarrow}(1-1/x)$, where $^\leftarrow$ denotes the left-continuous inverse function.
\end{lemma}
\begin{proof}[Proof of Lemma \ref{lemma:concentration:order:statistic}]

    Note that,  
$$
\Pr\suit{\abs{\frac{X^{(j)}_{n-k,n}  }{U_j(n/k)}-1}>\zeta} =  \Pr\suit{\frac{X^{(j)}_{n-k,n} }{U_j(n/k)}-1>\zeta}+   \Pr\suit{\frac{X^{(j)}_{n-k,n} }{U_j(n/k)}-1<-\zeta}.
$$

We first handle 
$$
\begin{aligned}
    \Pr\suit{\frac{X^{(j)}_{n-k,n} }{U_j(n/k)}-1>\zeta} =& \Pr\set{X^{(j)}_{n-k,n} >U_j(n/k) \suit{1+\zeta}  }\\
    =&\Pr\set{\sum_{i=1}^{n} \mI\suit{X_i^{(j)} > U_j(n/k) \suit{1+\zeta}} > k  } \\
    =& \Pr\set{\frac{1}{n}\sum_{i=1}^{n} \mI\suit{X_i^{(j)} > U_j(n/k) \suit{1+\zeta}}  > k/n  }. \\
\end{aligned}
$$
Denote 
$
Z_i^{(j)} = \mI\suit{X_i^{(j)} > U_j(n/k) \suit{1+\zeta}}.
$
Then, we have that, $|Z_i|\le 1$ and
$$
\begin{aligned}
    \bE (Z_i^{(j)}) = &\Pr\suit{X_i^{(j)} > U_j(n/k) \suit{1+\zeta}}=:\mu_n^{(j)},\\
    \textnormal{Var}(Z_i) =& \mu_n^{(j)}(1-\mu_n^{(j)}).
\end{aligned}
$$
By using Lemma \ref{Lemma:Bernstein}, we have that, 
$$
\begin{aligned}
    \Pr\suit{\frac{X^{(j)}_{n-k,n} }{U_j(n/k)}-1>\zeta} = & \Pr\suit{ \frac{1}{n}\sum_{i=1}^n Z_i - \mu_n^{(j)} > k/n-\mu_n^{(j)} } \\
    \le & 2\exp\suit{ -\frac{n(k/n-\mu_n^{(j)})^2}{2\mu_n^{(j)}(1-\mu_n^{(j)})+ 2(k/n-\mu_n^{(j)})/3} } \\
    =& 2 \exp\suit{ - k\frac{ k }{2n\mu_n^{(j)}} \frac{\suit{1-\frac{\mu_n^{(j)} n}{k}}^2}{1-\mu_n+ (k/(n\mu_n^{(j)}) -1)/3}  }\\
    =& 2 \exp\suit{ - k\frac{ 1 }{2\eta_n^{(j)}} \frac{\suit{1-\eta_n^{(j)}}^2}{1-\mu_n+ (1/\eta_n^{(j)} -1)/3}  },
\end{aligned}
$$
where $\eta_n^{(j)}:= \mu_n^{(j)} n/k$.

In the following, we give an upper bound and a lower bound of $\eta_n^{(j)}$.
Denote 
$$
\begin{aligned}
       \alpha(t)=&\max_{1\le j\le p}  \sup_{x>1} \abs{\frac{1-F_j(tx)}{1-F_j(t)}x^{1/\gamma_j}-1}, \\
       \alpha_0(t)=&\sup_{s\geq t}\alpha(s). 
\end{aligned}
$$
Then, we have that,
\begin{equation}\label{s:eq:uniform:F}
    \max_{1\le j\le p}  \sup_{x>1} \abs{\frac{1-F_j(tx)}{1-F_j(t)}x^{1/\gamma_j}-1}\le \alpha_0(t),
\end{equation}
and $\alpha_0(t)\to 0$, as $t\to\infty$, by using the condition
 \ref{condition:uniform:convergence}.

By  \eqref{s:eq:uniform:F}, we have that,  
$$
\abs{\frac{1-F_j(U_j(n/k)x)}{1-F_j(U_j(n/k))}x^{1/\gamma_j}-1 }\le  \alpha_0(U_j(n/k)),
$$
for all $x>1$ and $j=1,\dots,p$.  By using the condition \ref{condition:uniform:convergence}, we have that,  there exists $n_1 = n_1(\zeta, \max_{1\le j\le p}\gamma_j)$,  such that, for $n \ge n_1$, 
$$
\begin{aligned}
    \alpha_0(U_j(n/k))< \frac{1}{4}\frac{\zeta}{\gamma_j}<\frac{1}{8}, \quad j=1,\dots,p,
\end{aligned}
$$
by using the fact that 
$$
    \zeta< \frac{1}{2}\frac{\gamma_j}{1+\gamma_j}<\frac{1}{2}\gamma_j, \quad j=1,\dots,p.
$$
It follows  that,  for $n\ge n_1$,
\begin{equation}\label{s:eta:lower}
    \begin{aligned}
        \eta_n=  \frac{1-F_j\set{U_j(n/k)(1+\zeta)}}{1-F_j\set{U_j(n/k)}} 
        \ge & (1+\zeta)^{-1/\gamma_j}\set{1-\alpha_0(U_j(n/k))} \\
         \ge & (1-\zeta/\gamma_j)\set{1-\alpha_0(U_j(n/k))} \\
         \ge & \frac{7}{8} (1- \zeta/\gamma_j)\ge \frac{1}{4},
     \end{aligned}
\end{equation}
and  
\begin{equation}\label{s:eta:upper}
\begin{aligned}
\eta_n = \frac{1-F_j\set{U_j(n/k)(1+\zeta)}}{1-F_j\set{U_j(n/k)}} 
    \le &  (1+\zeta)^{-1/\gamma_j}  \set{1+\alpha_0(U_j(n/k))}\\
    \le & (1+\zeta)^{-1/\gamma_j} + \alpha_0(U_j(n/k)\\
    \le & 1-\frac{\zeta}{\gamma_j} +\frac{\zeta^2}{2\gamma_j} \suit{\frac{1}{\gamma_j}+1}+\alpha_0(U_j(n/k))\\
    \le & 1-\frac{3}{4}\frac{\zeta}{\gamma_j} +\alpha_0(U_j(n/k))\\
    \le &1-\frac{\zeta}{2\gamma_j}<1.
\end{aligned}
\end{equation}
Combining \eqref{s:eta:lower} and \eqref{s:eta:upper}, we obtain that,  for $n\ge n_1$,
$$
\frac{ 1 }{2\eta_n} \frac{\suit{1-\eta_n}^2}{1-\mu_n+ (1/\eta_n -1)/3}   \ge  \frac{1}{16\gamma_j^2}  \zeta^2,
$$
and hence
$$
\Pr\suit{\frac{X^{(j)}_{n-k,n} }{U_j(n/k)}-1>\zeta} \le  2 \exp\suit{  -\frac{k\zeta^2}{16\max_{1\le j\le p}\gamma_j^2} }.
$$

Similarly, we can show that,   for sufficiently large $n$,
$$
\Pr\suit{\frac{X^{(j)}_{n-k,n} }{U_j(n/k)}-1<-\zeta} \le 2 \exp\suit{  - \frac{k\zeta^2}{16\max_{1\le j\le p}\gamma_j^2} }.
$$
The proof is then complete.
\end{proof}

 \section{Additional Simulation}

 \subsection{Additional marginal distributions}
 In this subsection, we investigate the performance of the proposed clustering algorithm under models with
 Fr\'echet marginal distributions. Specifically, we consider the following two models:

{\noindent \bf Model (A-F)}.  The random variables $X^{(j)}, j=1,\dots, p$ are generated independently from  Fr\'echet distribution with shape parameter $1/\gamma_j$. The cumulative distribution function  of the  Fr\'echet distribution is given by 
 $$\text{Fr}_{1/\gamma_j}(x) = \exp(-x^{-1/\gamma_j}).$$

{\noindent \bf Model (B-F)}  We generate $(\widetilde{X}^{(1)},\dots, \widetilde{X}^{(p)})$ from a multivariate Cauchy distribution with scale matrix $\Sigma = (\sigma_{ij})_{i,j=1:p}$ and $\sigma_{ij} = 0.5^{|i-j|}$. Then we transform the marginal distribution of $\widetilde{X}_j$ to a  Fr\'echet distribution with shape parameter $1/\gamma_j$,  by 
$$
X^{(j)} = \text{Fr}_{1/\gamma_j}^{-1}\set{\text{St}_1(\widetilde{X}^{(j)})}.
$$

The simulation settings are consistent with those in Section \ref{sec:simulation}: sample size $n=2000$, number of clusters $g\in \set{3,4,5}$, $k = \lfloor 3\log^{1.05} (p) \rfloor$, $k^* = \lfloor n^{0.98}\rfloor$ and $\beta =  \min(2(k^*/k)^{-1}p+0.5,0.9)$.   
The average clustering accuracy is evaluated across varying values  of $q$ and $\Delta$, with the results illustrated in Figures  \ref{fig:q:fr} and \ref{fig:delta:fr}, respectively. The findings demonstrate that the proposed algorithm maintains strong performance under Fr\'echet marginal distributions, highlighting its robustness to different marginal distributional assumptions.

\begin{figure}[htb]
    \centering
    \includegraphics[width = 1\textwidth]{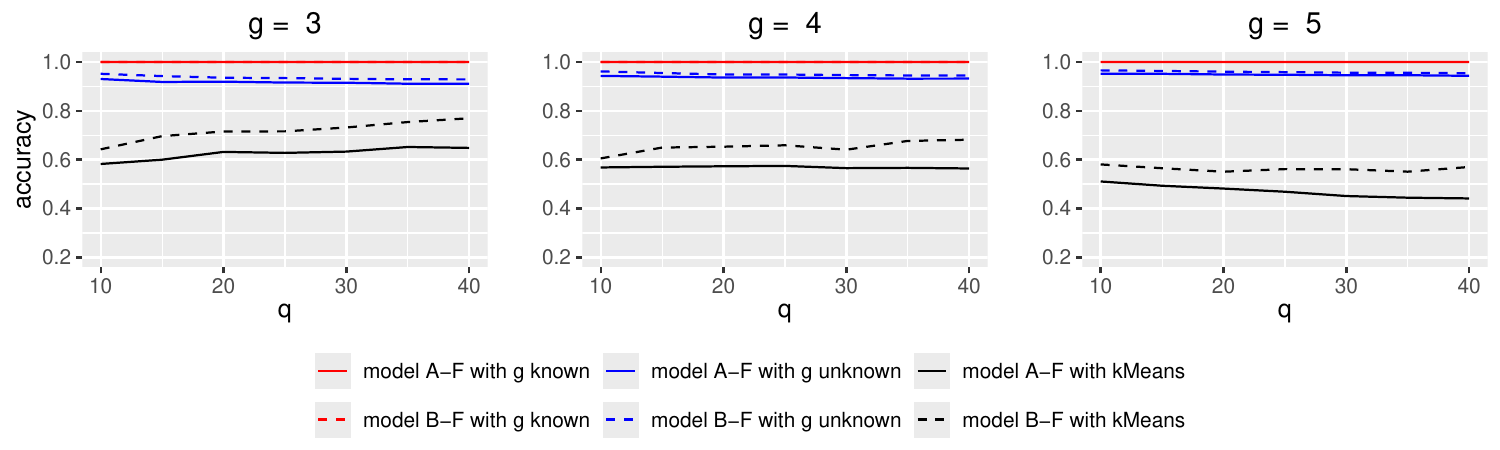}
    \caption{Clustering accuracy  against $q$ for   models (A-F) and (B-F). }
    \label{fig:q:fr}
\end{figure}

\begin{figure}[htb]
    \centering
    \includegraphics[width = 1\textwidth]{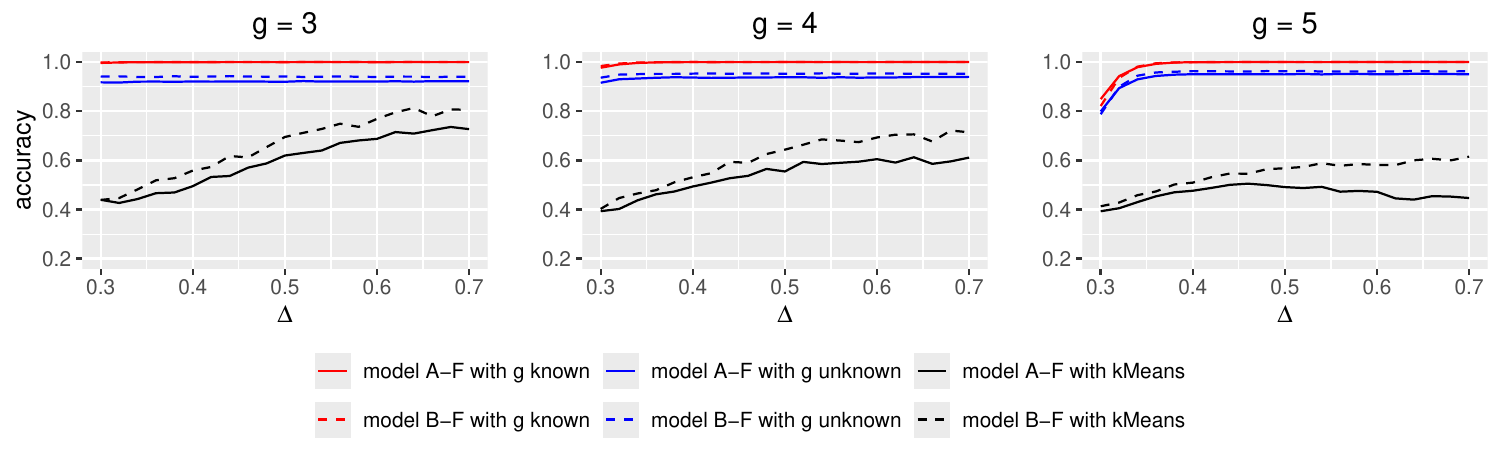}
    \caption{Clustering accuracy  against $\Delta$ for models (A-F) and (B-F). }
    \label{fig:delta:fr}
\end{figure}

 \subsection{Additional Dependence Structure}
 In this subsection, we examine the performance of the proposed clustering algorithm under alternative dependence structures. Specifically, we consider the following models:

 {\noindent \bf Model (C)}.  We generate $(\widetilde{X}^{(1)},\dots, \widetilde{X}^{(p)})$ from  a multivariate Cauchy distribution with scale matrix $\Sigma = (\sigma_{ij})_{i,j=1:p}$ and 
 $
 \sigma_{ij} = \mI(i=j) +0.5\mI(i\ne j)
 $. Then we transform the marginal distribution of $\widetilde{X}_j$ to an absolute Student-t distribution
 with degree of freedom $1/\gamma_j$,  by 
 $$
 X^{(j)} = \abs{\text{St}_{1/\gamma_j}^{-1}\set{\text{St}_1(\widetilde{X}^{(j)})}}.
 $$

{\noindent \bf Model (D)}  We generate $(\widetilde{X}^{(1)},\dots, \widetilde{X}^{(p)})$   from a multivariate Cauchy distribution with scale matrix $\Sigma = (\sigma_{ij})_{i,j=1:p}$ and $\sigma_{ij} = (-0.5)^{|i-j|}$. Then we transform the marginal distribution of $\widetilde{X}_j$ to an absolute Student-t distribution
with degree of freedom $1/\gamma_j$,  by 
$$
X^{(j)} = \abs{\text{St}_{1/\gamma_j}^{-1}\set{\text{St}_1(\widetilde{X}^{(j)})}}.
$$
The simulation settings remain unchanged: $n = 2000$, $ g\in \set{3,4,5}$,  $k = \lfloor 3\log^{1.05} (p) \rfloor$, $k^* = \lfloor n^{0.98}\rfloor$ and $\beta =  \min(2(k^*/k)^{-1}p+0.5,0.9)$. The average clustering accuracy is evaluated across varying values of $q$ and $\Delta$, with the results presented in Figures \ref{fig:q:CD} and \ref{fig:delta:CD}, respectively. The results indicate that the proposed algorithm is robust to different dependence structures.

\begin{figure}[htb]
    \centering
    \includegraphics[width = 1\textwidth]{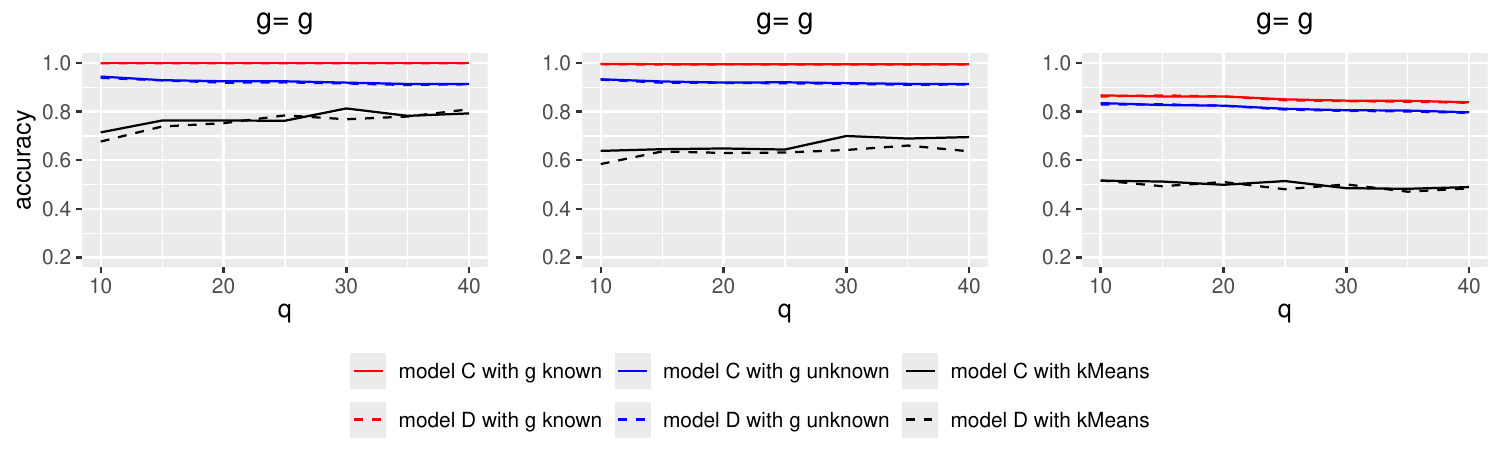}
    \caption{Clustering accuracy  against $q$ for models (C) and (D). }
    \label{fig:q:CD}
\end{figure}

\begin{figure}[htb]
    \centering
    \includegraphics[width = 1\textwidth]{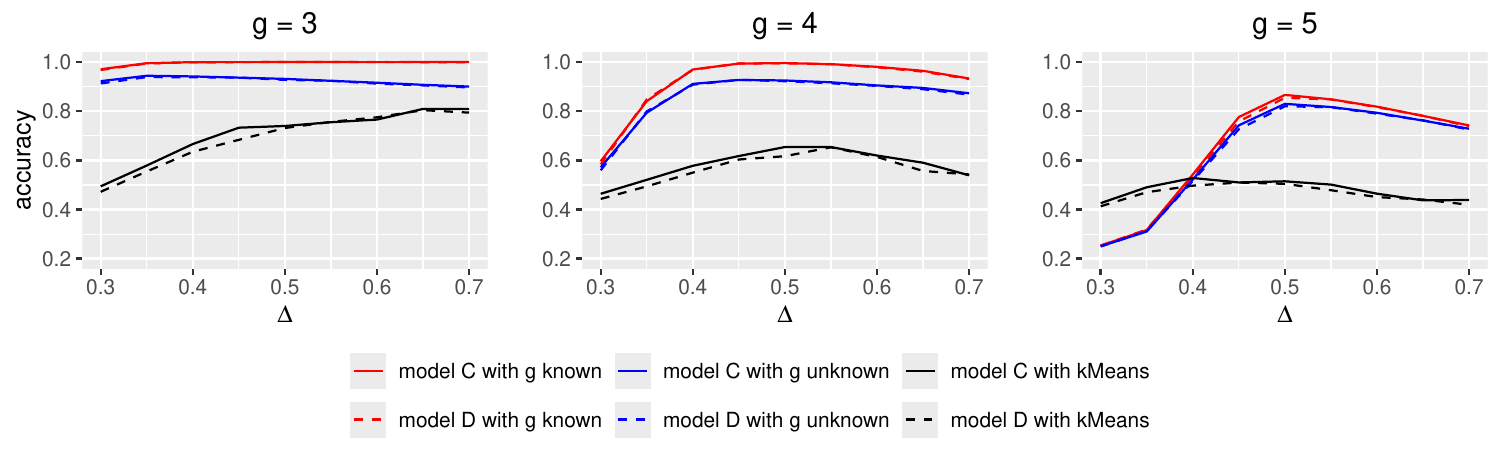}
    \caption{Clustering accuracy  against $\Delta$ for models (C) and (D). }
    \label{fig:delta:CD}
\end{figure}

\end{document}